\documentclass{article}

\usepackage{arxiv}

\usepackage[utf8]{inputenc} 
\usepackage[T1]{fontenc}    
\usepackage{hyperref}       
\usepackage{url}            
\usepackage{booktabs}       
\usepackage{amsfonts}       
\usepackage{nicefrac}       
\usepackage{microtype}      
\usepackage{lipsum}		
\usepackage{graphicx}
\usepackage{natbib}
\usepackage{doi}

\usepackage{shellesc}

\usepackage[draft]{minted}
\newcommand{\prolog}[1]{\mintinline[breaklines]{prolog}{#1}}

\usepackage{amssymb}
\usepackage{amsmath}
\usepackage{amsthm}
\newtheorem{definition}{Definition}
\newtheorem{example}{Example}
\newtheorem{theorem}{Theorem}

\title{Can AI expose tax loopholes? Towards a new generation of legal policy assistants}

\author{ \href{https://orcid.org/0000-0001-6378-5886}{\includegraphics[scale=0.06]{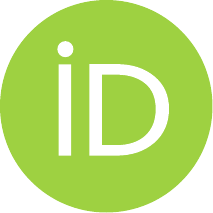}\hspace{1mm}Peter Fratrič} \\
	Department of Law and Tax\\
	HEC Paris, France\\
	\texttt{fratric@hec.fr} \\
	\And
	\href{https://orcid.org/0000-0002-0844-1391}{\includegraphics[scale=0.06]{orcid.pdf}\hspace{1mm}Nils Holzenberger} \\
	Data, Information, Graphs Lab\\
	Télécom Paris, Institut Polytechnique de Paris, France\\
	\texttt{nils.holzenberger@telecom-paris.fr} \\
	\AND
    \href{https://orcid.org/0000-0002-2841-2563}{\includegraphics[scale=0.06]{orcid.pdf}\hspace{1mm}David Restrepo Amariles} \\
	Department of Law and Tax\\
	HEC Paris, France\\
	\texttt{restrepo-amariles@hec.fr} \\
}

\renewcommand{\headeright}{Manuscript}
\renewcommand{\undertitle}{Manuscript}
\renewcommand{\shorttitle}{Can AI expose tax loopholes?}

\hypersetup{
pdftitle={A template for the arxiv style},
pdfsubject={q-bio.NC, q-bio.QM},
pdfauthor={David S.~Hippocampus, Elias D.~Striatum},
pdfkeywords={First keyword, Second keyword, More},
}

\begin{document}
\maketitle

\begin{abstract}
    The legislative process is the backbone of a state built on solid institutions. Yet, due to the complexity of laws~--- particularly tax law~--- policies may lead to inequality and social tensions. In this study, we introduce a novel prototype system designed to address the issues of tax loopholes and tax avoidance. Our hybrid solution integrates a natural language interface with a domain-specific language tailored for planning. We demonstrate on a case study how tax loopholes and avoidance schemes can be exposed. We conclude that our prototype can help enhance social welfare by systematically identifying and addressing tax gaps stemming from loopholes.
\end{abstract}

\section{Introduction}\label{sec:introduction}

Each year, billions of dollars in uncollected income and wealth taxes contribute to what is commonly referred to as the tax gap. For example, the global loss of tax revenue due to profit shifting remains to be about $10\%$ of corporate tax revenue collected \cite{alstadsaeter-et-all:report}. The lost revenue could otherwise be harnessed to alleviate poverty, fund critical public services such as education and healthcare, combat climate change, or promote economic investment; this is particularly true for low-income countries.\footnote{See United Nations' \textit{2024 Trade and development report} \url{https://unctad.org/system/files/official-document/tdr2024ch5_en.pdf}} This difference comes from two primary sources: policy design flaws, known as the \textit{policy gap}, and non-compliance due to tax evasion and avoidance, referred to as the \textit{compliance gap}. Notably,
both gaps
undermine public trust in societal institutions, which are essential to fostering sustainable economic development and progress \cite{acemoglu:institutional}.
Analyses point out that inefficient collection of taxes is the main driving force behind wealth inequality globally \cite{zucman-et-all:world,zucman:global}. 

\begin{figure}[h!]
    \centering
    \includegraphics[width=0.55\textwidth]{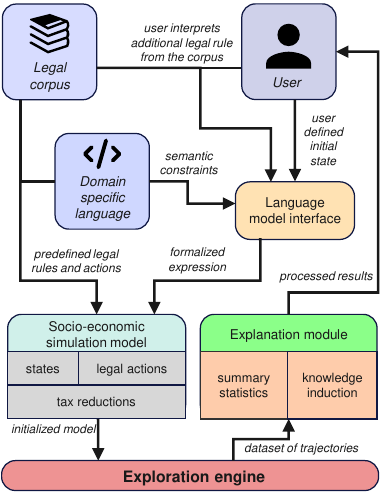}
    \caption{A subset of laws from the legal corpus are translated in a domain specific language (DSL), to define states, legal actions, and tax reductions in a socio-economic model. The user can define the initial state or interpret a new law from the corpus via a natural language interface constrained to produce DSL-compatible expressions. The exploration engine samples tax plans, conceptualized in the explanation module, and passed to the user.}
    \label{fig:figure1}
\end{figure}

In this study, we focus on tax loopholes, which originate in the policy gap (flawed laws), and aggravate the compliance gap, by generating opportunities to engage in tax avoidance ~\cite{blair-et-all:shelter}. 
Using concepts from agent-based modeling, we propose a prototype system described in Figure~\ref{fig:figure1}, and address the issues of tax policy conceptualized as research questions:
\begin{itemize}
    \item [RQ1] \textit{Availability of tax planning:} Can AI make tax planning broadly available?
    \item [RQ2] \textit{Democratization of tax law:} Can AI identify tax loopholes, allowing for a wider societal law evaluation?
    \item [RQ3] \textit{Improving social welfare:} Can AI assist policy design, thus improving aggregate measures of social welfare?
\end{itemize}

Loopholes often emerge due to misaligned regulations at both national and international levels, or through the misuse of well-intentioned tax incentives. We choose to illustrate our prototype on a notoriously known tax avoidance scheme, described in Example~\ref{ex:main}:
\begin{example}\label{ex:main}
    Consider the task of incorporating the child companies of a multinational company in various countries, to provide services derived from an intellectual property (IP). Naively, one could simply start a child company in each country, and sub-license IP rights to its subsidiary. However, roughly from 2000 to 2020, a notorious tax avoidance scheme called \textbf{Double Irish with the Dutch Sandwich} was used to avoid billions of dollars in taxes by utilizing two loopholes (as indicated in its name) and several international agreements.
\end{example}


\subsection{Related work}

Traditionally, tax policy was designed within the framework of optimal tax design \cite{mirrlees:taxdesign,casamatta:tax}. These principles were extended by the agent-based modeling paradigm, placing its focus on individual behavioral aspects rather than on analytic analysis. Several simulation models were developed (see e.g. \cite{garrido-et-al:tax}), with AI-economists being developed as an illustration of how a deep learning agent can be used to analyze and decide optimal tax policy \cite{zheng-et-al:aieconomicst}. These studies, while highly sophisticated, can be regarded as naive simplifications of the real world, because they exclude the possibility of clever non-compliance strategies. 
Tax gaps are often perceived as an issue of operationalization rather than of policy decision.
Several studies aimed to develop a model to include tax non-compliance, modeling it as, e.g., an external factor \cite{sandmo:ext}, a social network phenomenon \cite{korobow-et-al:tax}, or a matter of opportunity \cite{noguera-et-al:tax}). However, these studies address only the sociological aspect of tax avoidance and not the legal aspect that could prevent it. 

Although several detection models were developed \cite{zheng-et-al:attenet,cao-et-al:rrpu,wu-et-al:noveltax,polovnikov-et-al:ownership}, even state-of-the-art non-compliance detection methods continue to struggle with adversarial dynamics of continuously evolving non-compliance, representing a major challenge \cite{zang-et-al:dontignore,zhang-et-all:onlinecon,Ma-et-all:gnn}. Moreover, as an application of the principle of due process, any detection leading to prosecution or restrictions in a tax investigation must be conducted in a fair and explainable manner, further complicating efforts to combat tax non-compliance with artificial intelligence \cite{mehdiyev-et-al:xai}.
Only a handful of studies focused on pointing out possible loopholes in the tax law, that is, aiming to restrict non-compliant behavior \textit{ex ante} rather than \textit{ex post} \cite{khlif-et-all:tax}. The seminal work of \cite{hemberg-et-all:evo} shows that illicit tax avoidance schemes can exhibit inherent evolutionary structure. Our work integrates these strands of literature by utilizing simulations of potentially non-compliant behavior to guide the development of more effective tax policies.

In recent years, Large Language Models (LLMs) have shown they certainly have their place in future tools for policymaking and legal analysis, with impactful results demonstrated in legal question-answering systems \cite{huang-et-all:aila,bhambhoria-et-all:interpretable}, legal judgment prediction \cite{feng-et-all:legal,gan-et-all:judgment}, and information retrieval \cite{shao-et-all:bertpli}, \textit{inter alia}. A novel approach is taken in \cite{padhye:tax}, where the issue of tax compliance is viewed as a software testing problem, and by using a large language model one can point out possible inconsistencies in tax law. While these advancements show promise, several studies have highlighted significant limitations. In the legal domain, there is evidence that LLMs are insufficient in performing basic reasoning with legal statutes \cite{blairstanek23can}. 

Efforts to overcome these limitations in reasoning and planning~\cite{zhou24isrllm} have involved neuro-symbolic architectures, where LLMs are seen primarily as an interface between natural language and symbolic reasoning systems.
LLMs can be made to use APIs to query symbolic tools~\cite{schick23toolformer}, which can support both reasoning and planning. Using this approach, an LLM can be augmented with a planning algorithm, by letting an LLM generate a model of the world, further processed by a planner~\cite{guan23leveraging}. In this sense, generating a world model is similar to code generation~\cite{silver24generalized}, a task LLMs are generally well-suited for~\cite{xu22systematic}.
To connect LLMs to exact reasoning, the LINC model translates natural language into first-order logic and executes a theorem prover \cite{olausson-et-all:linc}. We draw inspiration from this approach and translate legal language into a domain-specific language, to serve as input to a formal planning and reasoning tool.


\subsection{Outline and contributions}

 
Section~\ref{sec:model} is focused on developing a domain specific simulation model with a language model input interface. This section contributes to a growing research in hybrid neuro-symbolic systems capable of planning. In Section~\ref{sec:simulation}, we show how the symbolic component of the developed model can be used to generate various incorporation plans, and point out certain tax loopholes that allow for a high amount of avoided taxes. The discussion on social implications of our technology is continued in Section~\ref{sec:policy}, where we briefly enter the framework of \textit{optimal tax design}, and show, both theoretically and in practice, how results of this study can be used to increase the value of a social welfare function. These results establish a connection between tax laws, compliance, and tax policy, addressing fundamental issues of optimal tax policy design.
We further discuss these implications along with future work in Section~\ref{sec:discussion}, and conclude in Section~\ref{sec:conclusion}.

\begin{table*}[bp]
\centering
\begin{tabular}{p{2.25cm}p{13.750cm}}
\toprule
reference code & associated legal documents \\
\midrule
EU-incorp  & Articles 49, 50(1) and (2)(g), and 54, second paragraph, of the Treaty on the Functioning of the European Union (TFEU) \\
USA-incorp  & Articles of Incorporation/Organization based on US federal law \\
BRM-incorp  & Companies Act 1981 \\
I-incorp & derived from case law: De Beers Consolidated Mines Ltd. v. Howe (1906) and Todd v. Egyptian Delta Land and Investment Co. Ltd. (1928) \\
transfer & common law understanding of IP transfer \\
license & common law understanding of IP licensing \\
2003/49/EC & EU Interest and Royalties Directive (2003/49/EC) \\
DCITA1969 & Dutch Corporate Income Tax Act (Wet op de vennootschapsbelasting 1969) \\
IRC-Sec162 & Internal Revenue Code (IRC) Section 162: Business Expense Deduction \\
2003/49/EC & EU Interest and Royalties Directive (2003/49/EC) \\
A8cNLctl1969 & Article 8c Dutch corporate tax law 1969 (Dutch: Wet op de vennootschapsbelasting 1969) \\
USA-wte & various US Treasury regulations and IRS provisions \\
check-the-box & IRC 7701, also known as Check-the-Box or CTB regulations, effective as of January 1, 1997 \\
\bottomrule
\end{tabular}
\caption{Legal corpus consisting of a selection of laws used as a reference for formalization of legal rules.}
\label{tab:appLaw}
\end{table*}

\section{Building a model: Ambiguity as a feature}\label{sec:model}

Here, we view tax calculations not as a precise formula to be identified and applied, but rather as a set of possible calculations differing only in the degree of their legal validity, as assessed by the legal community. To account for this ambiguity, it is necessary to design a system that allows to dynamically add new rules. For this reason, we build our model as a combination of static and dynamic formal concepts, the latter being defined by the user using a language model interface.

\subsection{Model and initial state}

In order to model the incorporation task, as outlined in Example~\ref{ex:main}, we define a transition system ${(S,A,\rightarrow,s_0)}$, where $S$ is the set of states, $A$ is the set of actions, $\rightarrow$ is the transition function, and $s_0$ the initial state.

\paragraph{States} We consider six state facts defining the state of incorporation of ${N-1}$ companies into ${M}$ countries, with the parent company defined in the initial state $s_0$. These state facts are:
\prolog{based(Company,Country)} defining operational country of a company, \prolog{managed(Company,Country)} if the country is managed from a different country, \prolog{isChildOf(Child,Parent)} defining corporate child-parent relationships, \prolog{ownsIP(Company,IP)} for countries that own IP, and \prolog{rentsIP(Owner,Renter,IP)} to define IP-licensing relationships between two companies. The initial state can be defined by the user, as illustrated in Figure~\ref{fig:init}. 

\begin{figure}
    \centering
    \includegraphics[width=0.7\linewidth]{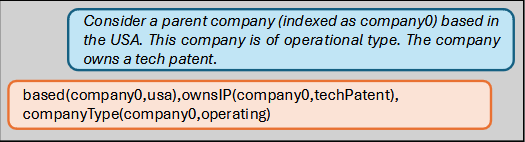}
    \caption{Illustration of state initialization by the user. Probability of obtaining the correct response can be found in Figure~\ref{fig:llmStats} for various language models.}
    \label{fig:init}
\end{figure}

\paragraph{Actions and transitions} State facts can be modified by actions. We consider three canonical actions, namely: \prolog{addChild} that instantiates a child company for a parent company, \prolog{rentIP} that instantiates sub-licensing contract between IP owner and IP licensee, and \prolog{transferIP} that instantiates the change of IP ownership contract. All these actions result in a modification of the state facts. Each action is actionable if its preconditions are satisfied. These are defined by corporate law for each jurisdiction. These laws are formalized into Prolog rules, based on the interpretation of a legal provision and labeled with a reference code \prolog{LegalRef}. 
These references are very important, because different jurisdictions or legal traditions can have a different understanding of a given action, which also allows us to track which laws were utilized in the process. For example, Europe and Great Britain have different incorporation traditions. While the location of a company is defined by its operational residence in the European tradition, it matters more from where the company is managed in the British tradition. This means the \prolog{addChild} action has different consequences for a different \prolog{LegalRef}. For a full list of legal references, see Table~\ref{tab:appLaw}.

\begin{figure}[t]
    \centering
    \includegraphics[width=0.7\linewidth]{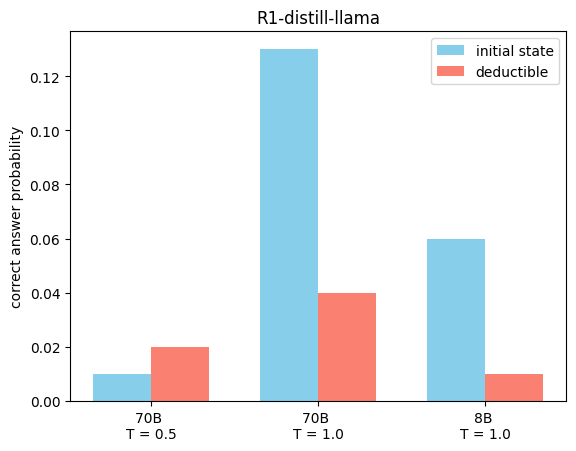}
    \caption{Simulation results of LLM sampling for the initial state prompt formalization (left bar) and deductible prompt formalization (right bar). Temperature value of the sampler was selected by grid search for values $[0.01, 0.02, 0.05, 0.1 , 0.2 , 0.5 , 1.0  ]$. For other values not included in the plot, no correct answer was found.}
    \label{fig:llmStats}
\end{figure}

\subsection{Economic context}

Each company can engage in commercial activity using its IP, or generate licensing revenues by instantiating a contract with another company. Both of these activities generate transactions of the form: 
\begin{small}
\begin{minted}{prolog}
transaction(Id,Time,Sender,Receiver,Amount)    
\end{minted}
\end{small}

\paragraph{Licensing contracts} Sub-licensing is modeled by a formal contractual agreement between two companies. The licensor can charge $90 \%$ of revenue from the licensee, including revenue from further sub-licensing with the same percentage; defined recursively. Transfer of IP generates a one-time transfer of funds. Since all companies are assumed to be controlled by the same entity, we assume that this transfer can occur between the parent company and any child company. To prevent cyclic transfer, we allow only one transfer. 

\paragraph{Commercial revenue} We assume only minimal economic context by defining a fixed revenue from commercial activity in each country. The countries we consider are:

\begin{table}[h!]
    \centering
    \begin{tabular}{c|cccc}
        \textbf{country} & U.S.A & Germany & Netherlands & Ireland \\ \hline
        \textbf{revenue} & \$700mil. & \$300mil. & \$100mil. & \$30mil. \\ 
    \end{tabular}
\end{table}
The necessary condition for revenue to be generated is for a company to own or license the IP in a country where the company is operationally present. We say that the incorporation process is \textit{multi-nationally complete} if and only if each company generates commercial revenue in every country above. We additionally consider Bermuda, but regard it as a tax haven only, of negligible economic significance in the simulation.

\subsection{Filing a tax return}

All transactions generated due to commercial activity or contractual agreements are taxed in a tax return in each country. Paying taxes is in principle a straightforward activity that requires multiplying the \textit{tax base} with the \textit{tax rate}. However, ambiguities are often contained in determining the exact value of the tax base and in choosing the right tax rate. This means that netto profits obtained by a company can be defined partially in a rigid way, while leaving some room for legal interpretation and structure of the legal argument for the tax return. The part subject to interpretation typically depends on deductibles and tax exemptions that may or may not apply for the transaction. Each tax reduction must be made in reference to a specific tax rule and applies if legal conditions are met, expressing it in a two-step process as: 
\begin{small}
\begin{minted}{python}
Input: ReductionRef, State, Base, Tax
Apply Reduction if Condition is True
TO_PAY = NewBase*NewTax
\end{minted}
\end{small}
If more than one reduction applies, we assume a rational choice of selecting the highest reduction. Modularity of tax reductions allows for a straightforward dynamic extension by the user. For example, the user can select a tax rule from the corpus, interpret it, and use the language interface to formalize it, as illustrated in Figure~\ref{fig:deduct}.

\begin{figure}
    \centering
    \includegraphics[width=0.75\linewidth]{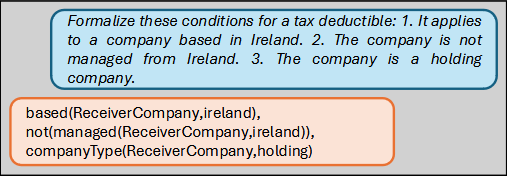}
    \caption{Illustration of introducing a new tax reduction rule by the user. The probability of obtaining the correct response can be found in Figure~\ref{fig:llmStats} for various language models.}
    \label{fig:deduct}
\end{figure}

\subsection{Utility function}

The process of incorporation together with the tax return allows us to define a utility function of the corporate structure. Let us denote $\Omega$ to be the set of trajectories, and ${\gamma\langle s_0,s \rangle \in \Omega}$ to be a trajectory from state $s_0$ to state $s$. For each trajectory ${\gamma \in  \Omega}$, we express the utility as:
\begin{equation}\label{eq:util}
    u(\gamma\langle s_0,s \rangle) = p(\gamma\langle s_0,s \rangle) - \phi(\gamma\langle s_0,s \rangle)
\end{equation}
where $\phi > 0$ denotes the incorporation costs, and $p$ are (netto) profits generated by the corporate structure after paying taxes, as described in the previous subsection. 

\section{Revealing tax loopholes}\label{sec:simulation}

Tax avoidance schemes typically involve careful planning that can utilize several tax rules to create an advantage. In order to identify legal loopholes, it is not sufficient to discover only the optimal sequence of steps. Rather, one needs to investigate tax planning holistically, and analyze suspicious frequencies of reference to tax rules. 

\subsection{Search space exploration}

Search space exploration and planning is a traditional area of AI and optimization, with several, both classical and state-of-the-art, algorithms developed. Algorithm selection can be highly dependent on the topology of the state transition system in question and on the intended goals. Our goal is to obtain a sample that includes both high utility trajectories, as well as trajectories of average utility, such that they can be contrasted. Since the main contribution of this study is not placed on specialized exploration algorithms, we consider algorithms that are general, simple, and align well with the design of our Prolog-Python environment.\footnote{For discussion on interface overhead, see \url{https://github.com/fratric/Rules2Lab}.} We chose randomized version of the \textit{best-first search} with state selection heuristic defined by:
\begin{equation}
\begin{split}
    P(s;t) &= \alpha_t P_{depth}(s;t) + (1 - \alpha_t) P_{utility}(s;t) \\
    P_{depth}(s;t) &\sim \exp{(-\frac{d_{G,t}(s,s_0)}{\tau_0})} \\ 
    P_{utility}(s;t) &\sim \exp{(u_t(s))}
\end{split}
\end{equation}
where ${P_{depth}(s;t)}$ is the probability distribution over expanded states, that inversely depends on depth of the state $s$ calculated as shortest-path distance $d_G$ from the initial state $s_0$ in the constructed search tree at step $t$, and ${P_{utility}(s;t)}$ is a probability distribution over expanded states at step $t$ that is proportional to the utility of the path ${\langle s_0,s \rangle}$. The time decay parameter, controlling sampling bias between ${P_{depth}}$ and ${P_{utility}}$, is defined as inverse time decay by the equation: ${\alpha_t = \frac{1}{1 + \beta t}}$ with the rate of decay $\beta$, that allows for control over the trade-off between exploration and exploitation. We set the rate to $\beta = 0.01$ with $t \in [0,50]$, expanding $1000$ states per iteration. Lastly, the \textit{termination criterion} is defined by the maximum amount of steps reached or by zero allowed actions available. 

\subsection{Legal summary statistics}


The dataset of sampled trajectories allows us to analyze quantities relevant for tax-loophole identification. First, one can select multi-nationally complete sequences, and order them by utility, obtaining a \textit{utility profile} of generated tax plans. Since selected sequences are economically equivalent, in the ideal case, the utility function should be (mostly) invariant with respect to the incorporation strategy and the structure of the tax return; meaning the profile should be (nearly) constant. If this is not the case, analysis can proceed by zooming on the statistical features of high-income tax plans. 

\begin{figure}[h]
    \centering
    \includegraphics[width=0.95\linewidth]{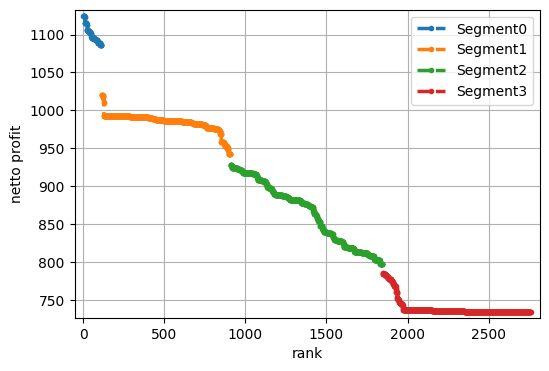}
    \caption{Segmentation of the utility profile based on the slope value. To obtain the segments, numerical differentiation of the curve was performed and peak detection algorithm with values peak size of $3.85$ and minimal distance between peaks equal to $100$ was used to identify points with high negative slope.}
    \label{fig:profile}
\end{figure}

\paragraph{Utility profile} On Figure~\ref{fig:profile}, we can observe that the utility profile of ordered multi-nationally complete incorporations is not resembling a constant function at all. Some trajectories are significantly more profitable that others, which can be attributed only to a better legal incorporation and tax strategy. Sharp steps in the profile indicate presence of different tax incorporation structures, i.e., various tax optimization schemes, which also indicates the presence of loopholes in the tax law. By numerically differentiating the curve, a peak detection algorithm can be employed to identify points that separate segments \cite{virtanen-et-all:scipy}. For a tax policy-maker, it is highly interesting to understand whether certain regularities occur among states that occupy the same segment. 

\begin{table*}[bp]
\centering
\begin{tabular}{lllll|l}
\toprule
 & segment 0 & segment 1 & segment 2 & segment 3 & type \\
\midrule
2003/49/EC & 1.670 & 3.854 & 0.962 & 4.066 & deductible \\
DCITA1969 & 2.946 & 0.211 & 0.598 & 0.099 &  \\
IRC-Sec162 & 1.223 & 1.248 & 2.246 & 1.164 & \\
\midrule
2003/49/EC & 1.670 & 3.854 & 0.962 & 4.066 & tax exemption \\
A8cNLctl1969 & 2.946 & 0.211 & 0.598 & 0.099 & \\
USA-wte & 1.223 & 1.248 & 2.246 & 1.164 & \\
check-the-box & 6.161 & 5.241 & 5.129 & 5.296 & \\
s23-IincpA & 3.982 & 0.325 & 2.563 & 1.080 & \\
\midrule
EU-inc & 0.250 & 0.332 & 0.271 & 0.333 & action \\
GB-inc & 0.177 & 0.088 & 0.142 & 0.091 &  \\
USA-inc & 0.000 & 0.000 & 0.000 & 0.000 &  \\
BMU-inc & 0.024 & 0.031 & 0.038 & 0.029 &  \\
license & 0.436 & 0.436 & 0.437 & 0.434 &  \\
transfer & 0.113 & 0.113 & 0.113 & 0.113 &  \\
\bottomrule
\end{tabular}
\caption{Frequency of paragraphs applied for each segment in the tax return for tax reductions, and paragraphs of corporate laws applied during the incorporation process. For tax reductions, we measure the amount of legal references applied per instance in a given segment. For corporate actions, we measure the frequency of applied references normalized by the amount of actions for each element of the segment. }
\label{tab:stats}
\end{table*}

\paragraph{Deductions and Tax Exemptions} Knowledge about utility segments allows us to ask whether certain tax base deductions or tax reductions were applied disproportionately often, pointing in the direction of potential loopholes in our legal system. In Table~\ref{tab:stats}, we can observe that certain laws were indeed exploited to reduce taxes. For example, the most profitable segment relies on two provisions of the Dutch tax law with reference codes DCITA1969 and A8cNLctl1969, that describe relationships between the Kingdom of Netherlands and certain low tax jurisdictions, such as Bermuda. This means that agreements between the Netherlands and offshore low tax jurisdictions, known as the \textit{Dutch sandwich}, are identifiable in the generated knowledge.

\paragraph{Incorporation process} Another useful piece of knowledge that can be extracted from the simulation results is the information about incorporation in various countries and tax jurisdictions. In order to do this, one needs to analyze trajectories without noise generated due to actions that were not necessary for the tax plan.
\begin{definition}\label{def:canonic}
    We say that a plan ${\gamma\langle s_0,s \rangle}$ is \textbf{canonical}, if for any ${\gamma' \neq \gamma}$ it holds that ${u(\gamma\langle s_0,s \rangle) > u(\gamma'\langle s_0,s \rangle)}$
\end{definition}
In the following theorem, we define conditions under which canonical representation exists. To do this, we introduce a notion of \textit{path independence}, identifying it with the equality ${p(\gamma\langle s_0,s \rangle) = p(s)}$. This means that the incorporation profits are evaluated only with respect to the final state of incorporation $s$.
\begin{theorem}
    Let the utility $u$ be defined as in equation (\ref{eq:util}), and let ${\phi(\gamma\langle s_0,s \rangle)}$ be proportional to the trajectory length ${||\gamma||}$. If ${p(\gamma\langle s_0,s \rangle)}$ is path-independent and ${||\gamma||}$ is minimal, then ${\gamma\langle s_0,s \rangle}$ is canonical.
\end{theorem}
\begin{proof}
Without loss of generality, let ${\phi(\gamma\langle s_0,s \rangle) = ||\gamma||}$. For any alternative path ${\gamma'\langle s_0,s \rangle}$ we can observe ${u(\gamma'\langle s_0,s \rangle) = p(s) - ||\gamma'|| < p(s) - ||\gamma|| = u(\gamma\langle s_0,s \rangle)}$ due to minimality of ${||\gamma||}$, hence ${\gamma\langle s_0,s \rangle}$ is canonical. 
\end{proof}
The importance of the result above lies not only in defining conditions under which incorporation strategies can be compared, but also reflects the reality of how humans understand the concept of tax avoidance schemes, that is, in their canonical forms.\footnote{In a way, as a normative concept \cite{field:taxonomy}.} Note, however, that noisy actions are introduced in practice to avoid detection, increasing their utility by performing camouflage (see e.g. \cite{slemrod:cheating}).

In Table~\ref{tab:stats}, one can see that in some segments, certain canonical representations overuse certain directives or laws in their tax returns. For example, the British tradition of incorporation is used in the top segment more often than in the other segments, but it can hardly be argued that this is sufficient evidence to identify the Double Irish loophole. This means the analysis on the statistical level should be followed by the structural level, which is performed in the next section.


\section{Policies: adjusting the tax design}\label{sec:policy}

The question of suitable tax policies is often answered in the literature as a task of finding a tax design function $T(z)$ for income $z$ that is maximizing the value of a social welfare function $W$, depending on presumed social and economic dynamics that generate a space of trajectories $\Omega$ for various tax designs. In this section, we demonstrate how inductive learning can be applied to increase the welfare by removing undesired tax optimization structures. 

\subsection{Optimal compliant tax design}

As discussed in Section~\ref{sec:introduction}, the classical tax design approach has a conceptual weakness: it does not account for interpretations of tax law, nor does it model its procedural aspects. This means that, in practice, the actual social welfare can be lower because of inefficient tax collection. We express this observation in the following definition:
\begin{definition}\label{def:ineff}
    A tax system is \textbf{operationally inefficient} with respect to a set $E$, if for a set of trajectories $\Omega$, where $E \subset \Omega$, it holds that 
    \begin{equation*}
        \delta(E) = W(\Omega \setminus E) - W(\Omega) > 0
    \end{equation*}
    We say the system is \textbf{operationally efficient} w.r.t. $E$ otherwise. 
\end{definition}
We can consider a set ${\{ t \in \Omega | H(t) \} }$, where $H$ is a boolean function, which we abbreviate by a slight abuse of notation as ${W(\Omega \setminus H)}$. In our case study, obtaining $H$ amounts to finding a collection of (boolean) rules that prohibit certain trajectories that are making the tax system less welfare-efficient. Defining a restriction $H$ is highly non-trivial, because the difference 
${\Delta(H) = \delta(H \cap E) - \delta(H \cap (\Omega \setminus E))}$ 
can be less than zero due to restrictions on compliant trajectories.

\subsection{Operational tax policy and induction}

The optimal selection $H$ can be modeled within the simulation environment as an induction problem. More specifically, an inductive-logic problem defined by the input tuple ${(E^+, E^-, B)}$, where $B$ is background knowledge and $E^+, E^-$ are disjoint sets of positive and negative examples, respectively. In our case, $E^+$ denotes instances of tax avoidance and $E^-$ of compliant behavior. The background knowledge is defined using $\Omega$. Positive and negative examples can be defined by thresholding the amount of netto profits within the environment, i.e., ${E^{+} = \{ \gamma \in \Omega | u(\gamma) > u_{+} \}}$ and ${E^{-} = \Omega \setminus E^{+}}$. While this provides us with an explicit rule defining which instances of behavior are non-compliant, a full access to sheltered income is, by definition, not available in practice. Therefore, in order to improve tax law, structural information about tax loopholes needs to be extracted. 

The optimal hypothesis $H$ in the ILP framework is defined by the requirements of \textit{consistency} and \textit{completeness}, meaning that ${\forall e \in E^+: H \cup B \models e}$ and ${\forall e \in E^-: H \cup B \nvDash  e}$, respectively.\footnote{For a complete formal treatment, see \cite{cropper-et-all:learning}.} We simplify the notation by writing $H \cap E^+$ for the former case and $H \cap E^-$ in the latter case, interpreting $H$ as a set defined by boolean rules that intersect sets $E^+$ and $E^-$ defined by element listing. 
The following theorem links the optimal solution of the inductive programming problem to operational efficiency. 
\begin{theorem}\label{th:ind}
     Let the induction problem ${(E^+, E^-, B)}$ be defined using the set $\Omega$ of operationally inefficient tax system with respect to $E^+$. If the optimal solution $H^*$ of the induction problem exists, then ${\Delta(H^*) > 0}$. 
\end{theorem}
\begin{proof}
    For any $H$, ${\Delta(H) = \delta(H \cap E^+) - \delta(H \cap E^-)}$. By consistency of $H^*$, we have ${H^* \cap E^- = \varnothing}$, thus for $H = H^*$ it holds that ${\delta(H^* \cap E^-) = \delta(\varnothing) = 0}$. By completeness of $H^*$, we have ${H^* \cap E^+ = E^+}$, therefore ${\delta(H^* \cap E^+) = \delta(E^+)}$ with ${\delta(E^+) > 0}$ due to operational inefficiency w.r.t $E^+$. This means $\Delta(H^*) > 0$.
\end{proof}

Theorem \ref{th:ind} states an intuitive and powerful guarantee of how social welfare can be increased by identifying strategies of tax payers that underpay, and this statement is made independently of how we define the social welfare function or the tax design. We discuss its limitations in Section~\ref{sec:discussion}.

\subsection{Suboptimal solution}

Finding the optimal hypothesis might be computationally intractable, as in realistic cases it would require to search through a large hypothesis space. Moreover, the optimal hypothesis might not exist due to poor selection of the profit threshold, which means all strategies contained in one segment might not have a uniquely identifiable common structural property. For this reason, one needs to select an induction algorithm that is robust to noise and can provide sub-optimal inductive explanations. We choose \texttt{PyGol}, an abductive/inductive logic programming tool based on meta inverse entailment that combines efficiency and robustness \cite{varghese-et-all:pygol}. Focusing only on multi-nationally complete trajectories, we define the head of our rule as \prolog{taxScheme(A,B,C,D)} where \prolog{A} is a company based in Ireland, \prolog{B} in the Netherlands, \prolog{C} in the USA, and \prolog{D} in Germany. Since we have predefined what variable corresponds to operational base in each respective country, we restrict the background knowledge only to literals \prolog{ownsIP}, \prolog{managed}, and \prolog{rentsIP}, which provides essential information about the sub-licensing structure synthesized among the companies. We run the \texttt{PyGol} algorithm for a maximum of $7$ literals. Obtained solution $H$ has accuracy $0.981$, precision $0.833$, specificity $0.670$, sensitivity $0.994$, and F1-score $0.743$. Learned rules of the hypothesis $H$ might not be easily interpretable for a legal expert, because the information about the tax avoidance scheme might be lost in the output complexity. To make the result easier to understand, the terms used for each rule of $H$ can be weighted by their individual F1-score and plotted as a knowledge graph, as depicted on Figure~\ref{fig:scheme}. This knowledge graph extracts the operational essence of the Double Irish Dutch Scheme.

\begin{figure}
    \centering
    \includegraphics[width=1\linewidth]{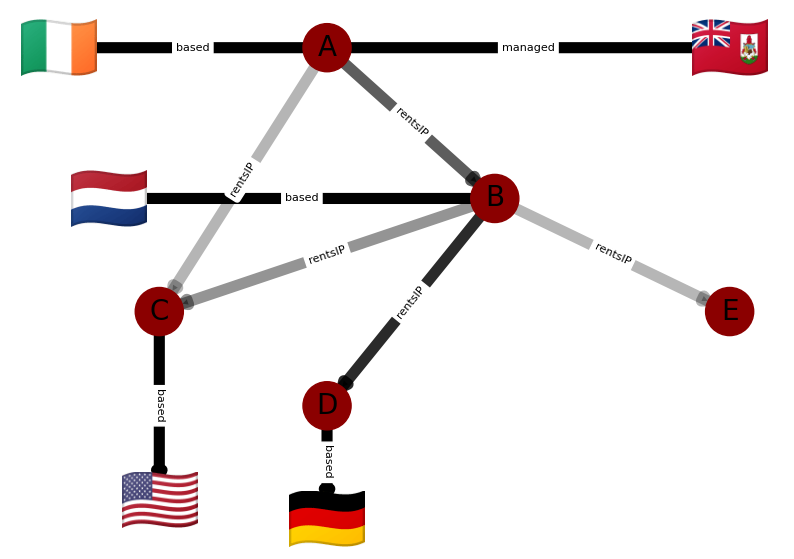}
    \caption{Visualization of first $10$ edges of hypothesis $H$ weighted by F1-score. Ireland based company A is managed from Bermuda, and rents IP to company B based in the Netherlands, which rents IP further to companies E, D, and C. This outlines the essence of the Double Irish with the Dutch Sandwich scheme.}
    \label{fig:scheme}
\end{figure}

\section{Discussion and future work}\label{sec:discussion}

\paragraph{Answering RQ1} We have demonstrated that, in principle, it is possible to conduct tax planning even in the very advanced case of Example~\ref{ex:main}. In order to make tax planning broadly available even for people with limited resources, the main bottleneck is the formalization of legal rules using language models, because even larger models do not perform sufficiently well. Leveraging a corpus of human-formalized rules and employing fine-tuning methods can be considered as a potential solution.

\paragraph{Answering RQ2} Using Example~\ref{ex:main} as a case study, we have successfully demonstrated that tax loopholes can be exposed. If our approach is employed by government or non-government bodies, with a sufficient amount of computing power, then every new law can be publicly evaluated. Since our approach allows for freedom of interpretation, various formalizations can be added into the knowledge base, simulated, and compressed results stored in a public database. Developing a minimal architecture that allows for a large scale open-source adaptation of our prototype that can be used by public institutions is our main future goal.

\paragraph{Answering RQ3} We have demonstrated how a statistical and structural description of a tax avoidance scheme can be made using our prototype system. We believe our system can act as a helpful policy assistant, providing essential information about possible tax loopholes. Theorem \ref{th:ind} provides theoretical guarantees supported by practical methods on how to improve social welfare. Still, we argue that caution should be taken when formulating operational restrictions. Firstly, the solution $H^*$ is obtained from a finite sample of trajectories drawn from one particular legal interpretation, which may introduce computational and modeling errors. Second, due to these errors, one should not aim to restrict all behaviors that can be regarded as non-compliant, but focus only on a small subset of the most severe instances while keeping a level of acceptable low-income non-compliance. We do not regard these inaccuracies as design flaws, but rather as defining aspects of legal decision-making. Lastly, our solution did not propose changes to specific laws, as we believe the final decision should be left for the human policy-maker. Nonetheless, extending our methodology by including a language model capable of processing the simulation results and proposing legal reformulations is among our long-term goals. 

\section{Conclusions}\label{sec:conclusion}

This study holds a powerful promise for the legal community, as well as for the policy-makers and general public. While certain components of our prototype appear to be the correct design choice, the issue of efficient, dynamic, and user-friendly formalization of legal rules persists as the greatest challenge. Extending our prototype with higher numbers of formalized rules and running more scenarios has the potential to result in publicly exposing unknown tax avoidance schemes for the benefit of society.




\bibliographystyle{unsrtnat}
\bibliography{references}  






\end{document}